\documentclass[11pt]{article}
\usepackage{microtype}
\usepackage{graphicx}
\usepackage{fullpage}
\usepackage{amsmath,amsthm,amssymb}
\usepackage{color}
\usepackage{verbatim}

\title{Algorithmic interpretations of fractal dimension}

\author{
Anastasios Sidiropoulos\footnote{Dept.~of Mathematics and Dept.~of Computer Science \& Engineering, The Ohio State University.}
\footnote{Supported by the National Science Foundation (NSF) under grant CCF 1423230 and award CAREER 1453472.}
\and
Vijay Sridhar\footnote{Dept.~of Computer Science \& Engineering, The Ohio State University.}
\footnotemark[2]
}
\date{}

\begin{document}

\def\lf{\left\lfloor}   
\def\rf{\right\rfloor}
\def\lc{\left\lceil}   
\def\rc{\right\rceil}
\newtheorem{theorem}{Theorem}
\newtheorem{proposition}[theorem]{Proposition}
\newtheorem{claim}[theorem]{Claim}
\newtheorem{lemma}[theorem]{Lemma}
\newtheorem{corollary}[theorem]{Corollary}
\newtheorem{definition}[theorem]{Definition}
\newtheorem{observation}[theorem]{Observation}
\newtheorem{fact}[theorem]{Fact}
\newtheorem{property}{Property}
\newtheorem{notation}{Notation}[section]
\newtheorem{algorithm}{Algorithm}
\newtheorem{conjecture}{Conjecture}
\newtheorem{question}[conjecture]{Question}

\newcommand{\eps}{\varepsilon}
\newcommand{\tw}{\mathsf{tw}}
\newcommand{\pw}{\mathsf{pw}}
\newcommand{\dimF}{\dim_{\mathsf{f}}}
\newcommand{\dimD}{\dim_{\mathsf{d}}}
\newcommand{\diam}{\mathsf{diam}}
\newcommand{\dist}{\mathsf{dist}}
\newcommand{\volume}{\mathsf{volume}}
\newcommand{\ball}{\mathsf{ball}}
\newcommand{\sphere}{\mathsf{sphere}}
\newcommand{\origin}{\mathbf{o}}
\newcommand{\OPT}{\mathsf{OPT}}
\newcommand{\XXX}{\textcolor{red}{XXX}}
\newcommand{\siep}{Sierpi\'{n}ski}
\newcommand{\ImSucc}{S}
\newcommand{\Succ}{S'}
\newcommand{\shrunk}{\mathsf{shrunk}}
\newcommand{\rep}{\mathsf{rep}}
\newcommand{\length}{\mathsf{length}}
\newcommand{\size}{\mathsf{size}}
\newcommand{\father}{\mathsf{father}}
\newcommand{\Near}{\mathsf{Near}}

\newcommand{\CC}[1]{\textcolor{red}{c_{#1}}}

\maketitle

\begin{abstract}
We study algorithmic problems on subsets of Euclidean space of low \emph{fractal dimension}. These spaces are the subject of intensive study in various branches of mathematics, including geometry, topology, and measure theory. 
There are several well-studied notions of fractal dimension for sets and measures in Euclidean space.
We consider a definition of fractal dimension for finite metric spaces which agrees with standard notions used to empirically estimate the fractal dimension of various sets.
We define the fractal dimension of some metric space to be the infimum $\delta>0$, such that for any $\eps>0$, for any ball $B$ of radius $r\geq 2\eps$, and for any $\eps$-net $N$ (that is, for any maximal $\eps$-packing), we have $|B\cap N|=O((r/\eps)^\delta)$.

Using this definition we obtain faster algorithms for a plethora of classical problems on sets of low fractal dimension in Euclidean space.
Our results apply to exact and fixed-parameter algorithms, approximation schemes, and spanner constructions.
Interestingly, the dependence of the performance of these algorithms on the fractal dimension nearly matches the currently best-known dependence on the standard Euclidean dimension.
Thus, when the fractal dimension is strictly smaller than the ambient dimension, our results yield improved solutions in all of these settings.

We remark that our definition of fractal dimension is equivalent up to constant factors to the well-studied notion of doubling dimension.
However, in the problems that we consider, the dimension appears in the exponent of the running time, and doubling dimension is not precise enough for capturing the best possible such exponent for subsets of Euclidean space.
Thus our work is \emph{orthogonal} to previous results on spaces of low doubling dimension;
while algorithms on spaces of low doubling dimension seek to extend results from the case of low dimensional Euclidean spaces to more \emph{general} metric spaces, our goal is to obtain faster algorithms for \emph{special} pointsets in Euclidean space.

\end{abstract}


\section{Introduction}
Sets of non-integral dimension are ubiquitous in nature and can be used to model a plethora of processes and phenomena in science and engineering \cite{takayasu1990fractals}.
Sets and measures in Euclidean space of certain fractal dimension are the subject of study in several branches of mathematics, including geometry, topology, and measure theory.

In many problems in computational geometry, the dimension of the input set often determines the complexity of the best-possible algorithms.
In this work we study the computational complexity of geometric problems on sets of bounded fractal dimension in low-dimensional Euclidean space.
We observe the following interesting phenomenon: For many problems, it is possible to obtain algorithms with dependence on the fractal dimension similar to the best-possible dependence to the standard Euclidean dimension.
This implies asymptotically faster algorithms when the fractal dimension of the input is smaller than the ambient dimension.

\subsection{Definition of fractal dimension}
Intuitively, some $X\subseteq \mathbb{R}^d$ has fractal dimension $\delta\in [0,d]$ if when scaling $X$ by a factor of $\alpha>0$, the ``volume'' of $X$ is multiplied by a factor $\alpha^\delta$.
There are many different ways this intuition can be formalized, such as Hausdorff dimension, Minkowski dimension, and so on.
Unfortunately, some of these definitions are not directly applicable in the context of discrete computational problems.
For example, the Hausdorff dimension of any countable set is 0.

Despite this, there are some natural methods that are used to estimate the fractal dimension of a set in practice.
Let $X\subseteq \mathbb{R}^d$.
Let $\Gamma_\eps$ be a $d$-dimensional grid where each cell has width $\eps>0$, and let $I_\eps(X)$ be the number of cells in $\Gamma_\eps$ that intersect $X$.
The \emph{fractal box-counting dimension} of $X$ is defined to be $\lim_{\eps\to 0}\log(I_\eps(X))/\log(1/\eps)$ \cite{falconer2004fractal}.
This definition is often used experimentally as follows:
Intersect $X$ with a regular lattice $(\eps \mathbb{Z})^d$,  and estimate the rate by which the cardinality of the intersection grows when $\eps \to 0$.
In that context, $X$ has fractal dimension $\delta$ when the size of the intersection grows as $(1/\eps)^\delta$
\cite{khoury2010fractal}.

We consider a definition that is closely related to box-counting dimension, but is more easily amenable to algorithmic analysis.
Let $S\subseteq A$. We say that $S$ is an \emph{$\eps$-covering} of $A$ if for any $x \in A$ we have that $\dist(\{x\},S)\leq \eps$.
For any $x \in A$ and $y \in S$ we say that $x$ is covered by $y$ if $\rho(x,y) \leq \eps$. $S$ is an \emph{$\eps$-packing} if for any $x,y \in S$ we have $\rho(x,y) \geq \eps$. If $S$ is both an $\eps$-covering and an $\eps$-packing of $A$ then we say that $S$ is an $\eps$-net of $A$. We define the \emph{fractal dimension} of some family of pointsets $P\subseteq \mathbb{R}^d$, denoted by $\dimF(P)$, to be the infimum $\delta$, such that for any $\eps>0$ and $r\geq 2\eps$, for any $\eps$-net\footnote{We arrive at an equivalent definition if we require $N$ to be a $\eps$-packing instead of a $\eps$-net.} $N$ of $P$, and for any $x\in \mathbb{R}^d$, we have 
$|N \cap \ball(x,r)| = O((r/\eps)^\delta)$.
For the sake of notational simplicty, we will be referring to the fractal dimension of some familty of pointsets $P$, as the fractal dimension of the pointset $P$, with the understanding that in the asymptotic notation $|P|$ is unbounded.

Figure \ref{fig:carpet} depicts an example of an infinite family of discrete pointsets $P$ with non-integral fractal dimension constructed as follows:
We begin with the $3^k\times 3^k$ integer grid, for some $k\in \mathbb{N}$, we partition it into 9 subgrids of equal size, we delete all the points in the central subgrid, and we recurse on the remaining 8 subgrids.
The recursion stops when we arrive at a subgrid containing a single point.
This is a natural discrete variant of the Sierpi\'{n}ski carpet.
It can be shown that $\dimF(P)=\log_{3}8$, which is equal to the Hausdorff dimension of the standard Sierpi\'{n}ski carpet.

\begin{figure}
\begin{center}
\scalebox{0.4}{\includegraphics{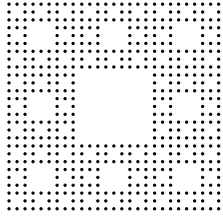}}
\caption{A discrete variant of the Sierpi\'{n}ski carpet for $k=3$.\label{fig:carpet}}
\end{center}
\end{figure}

\subsection{Why yet another notion of dimension?}
We now briefly compare the above notion of fractal dimension to previous definitions and motivate its importance.
The most closely related notion that has been previously studied in the context of algorithm design is \emph{doubling dimension} \cite{assouad1983plongements,heinonen2012lectures,gupta2003bounded}.
We recall that the doubling dimension of some metric space $M$, denoted by $\dimD(M)$, is defined to be $\log \kappa$, where $\kappa$ is the minimum integer such that for all $r>0$, any ball in $M$ of radius $r$ can be covered by at most $\kappa$ balls of radius $r/2$.
It is easy to show that for any metric space $M$, we have\footnote{Note that for a set $X$ containing two distinct points we have $\dimF(X)=0$ while $\dimD(X)=1$ and thus it is not always the case that $\dimD(X)=O(\dimF(X))$.}
$\dimD(M) = \dimF(M) + O(1)$ and $\dimF(M) = O(\dimD(M))$.
Thus our definition is equivalent to doubling dimension up to constant factors.
However, in the problems we consider, the dimension appears in the \emph{exponent} of the running time of the best-known algorithms;
therefore, determining the best-possible constant is of importance.
As we shall see, for several algorithmic problems, our definition yields nearly optimal bounds on this exponent, while doubling dimension is not precise enough for this task.

Let us illustrate this phenomenon on the problem of solving TSP exactly of a set of $n$ points in the Euclidean plane.
It is known that TSP admits an algorithm with running time $2^{O(\sqrt{n} \log{n})} n^{O(1)}$ in this case \cite{smith1998geometric}.
Moreover, the exponent of $O(\sqrt{n} \log{n})$ is known to be nearly optimal assuming the Exponential Time Hypothesis (ETH) \cite{marx2014limited} (see later in this Section for a more precise statement). 
We show that for sets of fractal dimension $\delta\in (1,2]$, there exists an algorithm with running time $2^{O(n^{1-1/\delta} \log n)}$.
Thus, for any fixed $\delta<2$, we achieve an asymptotically faster algorithm than what is possible for general pointsets (assuming ETH).
On the other hand, it is known that the unit disk cannot be covered with 6 disks of radius $1/2$ (see \cite{zahn1962black}). Thus $\dimD(\mathbb{R}^2) \geq \log_2 7 > 2.807$, while $\dimF(\mathbb{R}^2)=2$.
Therefore doubling dimension is not precise enough to capture the best-possible exponent in this setting.

In summary, while algorithms on spaces of low doubling dimension seek to extend results from the case of low dimensional Euclidean space to a more \emph{general} setting, our goal is to obtain faster algorithms for \emph{special} classes of pointsets in Euclidean space.

\subsection{Our results}
We obtain algorithms for various problems on sets of low fractal dimension in Euclidean space.
We consider exact algorithms, fixed parameter algorithms, and approximation schemes. 
In each one of these settings, we pick classical representative problems.
We believe that our techniques should be directly applicable to many other problems.

\textbf{Exact algorithms.}
We first consider exact algorithms in $\mathbb{R}^d$.
It is known that for any fixed $d$, TSP on a set of $n$ points in $\mathbb{R}^d$ can be solved in time $2^{O(n^{1-1/d} \log n)}$ \cite{smith1998geometric}.
By adapting ideas from the Euclidean setting,
we show that TSP on a set of $n$ points of fractal dimension $\delta > 1$ in constant-dimensional Euclidean space, can be solved in time $2^{O(n^{1-1/\delta} \log n)}$.
When $\delta=1$ and $\delta<1$, our algorithm has running time $n^{O(\log^2 n)}$ and $n^{O(\log n)}$ respectively.
We remark that it has been shown by Marx and Sidiropoulos \cite{marx2014limited} that assuming ETH, there is no algorithm for TSP in $\mathbb{R}^d$ with running time $2^{O(n^{1-1/d-\eps})}$, for any $\eps>0$.
Thus, our result bypasses this lower bound for sets of low fractal dimension.
In particular, our result implies that, in a certain sense, the hardest instances for TSP in $\mathbb{R}^d$ must be close to \emph{full-dimensional}; that is, they must have fractal dimension close to $d$.
Our technique also extends to the Minimum Rectilinear Steiner Tree problem in $\mathbb{R}^2$.

\textbf{Parameterized problems.}
We also consider algorithms for problems parameterized by the value of the optimum solution.
A prototypical geometric problem in this setting is Independent Set of unit balls in $\mathbb{R}^d$.
Formally, we show that given a set $D$ of unit balls in $\mathbb{R}^d$, the $k$-Independent Set problem on $D$ can be solved in time  $n^{O(k^{1-1/\delta})}$, for any fixed $d$, where $\delta>1$ is the fractal dimension of the set of centers of the disks in $D$.
When $\delta \leq 1$, we get an algorithm with running time $n^{O(\log k)}$.
Previously known algorithms for this problem in $d$-dimensional Euclidean space have running time $n^{O(k^{1-1/d})}$, for any $d\geq 2$  \cite{alber2002geometric,marx2014limited}.
Moreover, it has been shown that there is no algorithm with running time $f(k) n^{o(k^{1-1/d})}$, for any computable function $f$, assuming ETH \cite{marx2014limited} (see also \cite{marx2005efficient}).
Thus, our result implies that this lower bound can also be bypassed for sets of fractal dimension $\delta < d$.

\textbf{Approximation schemes.}
We next consider approximation schemes.
Let $P$ be a set of $n$ points of fractal dimension $\delta>0$, in $d$-dimensional Euclidean space.
We show that for any $R>0$, for any $\ell>0$, we can compute a $(1+d/\ell)$-approximate $R$-cover of $P$ in time $\ell^{d+\delta} n^{O({(\ell \sqrt{d})^{\delta})}}$.
This matches the performance of the algorithm of Hochbaum and Maass \cite{hochbaum1985approximation} after replacing $\delta$ by $d$.
We also obtain a similar algorithm for the $R$-packing problem.

\textbf{Spanners and pathwidth.}
Recall that for any pointset in $\mathbb{R}^d$, and for any $c\geq 1$, a \emph{$c$-spanner} for $P$ is a graph $G$ with $V(G)=P$, such that for all $x,y\in P$, we have
$\|x-y\|_2 \leq d_G(x,y) \leq c \cdot \|x-y\|_2$,
where $d_G$ denotes the shortest path distance in $G$.
The parameter $c$ is called the \emph{dilation} of $G$.
It is known that for any $\eps>0$, any set of $n$ points in $\mathbb{R}^d$ admits a $(1+\eps)$-spanner of size $n (1/\eps)^{O(d)}$ \cite{salowe1991construction,vaidya1991sparse}.
We strengthen this result in the following way.
We show that for any $\eps>0$, any set of $n$ points of fractal dimension $\delta$ in constant-dimensional Euclidean space admits a $(1+\eps)$-spanner of size $n(1/\eps)^{O(d)}$, and of pathwidth at most $O(n^{1-1/\delta} \log n)$ if $\delta>1$, at most $O(\log^{2} n)$ if $\delta= 1$, and at most $O(\log n)$ if $\delta<1$.
Our spanner is obtained via a modification of the construction due to Vaidya \cite{vaidya1991sparse}.
This provides a general polynomial-time reduction for geometric optimization problems on Euclidean instances of low fractal dimension to corresponding graph instances of low pathwidth.
This result can be understood as justification for the fact that instances of low fractal dimension appear to be ``easier'' than arbitrary instances.
We remark that our construction also implies, as a special case, that arbitrary $n$-pointsets in $\mathbb{R}^d$ admit $(1+\eps)$-spanners of size $n(1/\eps)^{O(d)}$ and pathwidth $O(n^{1-1/d} \log n)$; this bound on the pathwidth appears to be new, even for the case $d=2$.

\subsection{Related work}
There is a large body of work on various notions of dimensionality in computational geometry.
Most notably, there has been a lot of effort on determining the effect of doubling dimension on the complexity of many problems \cite{har2006fast,bartal2012traveling,cole2006searching,krauthgamer2005black,gottlieb2008optimal,krauthgamer2005measured,chan2009small,chan2005hierarchical,gupta2012online,talwar2004bypassing}.
Other notions that have been considered include low-dimensional negatively curved spaces \cite{krauthgamer2006algorithms}, 
growth-restricted metrics \cite{karger2002finding},
as well as 
generalizations of doubling dimension to metrics of  bounded global growth \cite{hubert2012approximating}.

A common goal in all of the above lines of research is to extend tools and ideas from the Euclidean setting to more general geometries.
In contrast, as explained above, we study restricted classes of Euclidean instances, with the goal of obtaining faster algorithms than what is possible for arbitrary Euclidean pointsets.

\subsection{Notation and definitions}
Let $(X,\rho)$ be some metric space.
For any $x\in X$ and $r\geq 0$, we define $\ball(x,r)= \{ y \in X : \rho(x,y) \leq r \}$ and $\sphere(x,r)= \{ y \in X : \rho(x,y) = r \}$.
For some $A,B\subseteq X$, we write $\dist(A,B)=\inf_{x\in A, y\in B}\{ \rho(x,y)\}$.
For some $r\geq 0$, we write 
$N(A,r) = \{x\in X : \dist(A, \{x\}) \leq r\}$.
Let $S\subseteq A$. We say that $S$ is an \emph{$\eps$-covering} of $A$ if for any $x \in A$ we have that $\dist(\{x\},S)\leq \eps$.
For any $x \in A$ and $y \in S$ we say that $x$ is covered by $y$ if $\rho(x,y) \leq \eps$. $S$ is an \emph{$\eps$-packing} if for any $x,y \in S$ we have $\rho(x,y) \geq \eps$. If $S$ is both an $\eps$-covering and an $\eps$-packing of $A$ then we say that $S$ is an $\eps$-net of $A$.

We recall the following definition from \cite{smith1998geometric}.
Let $D$ be a collection of subsets of $\mathbb{R}^d$. $D$ is said to be \emph{$\kappa$-thick} if no point is covered by more than $\kappa$ elements of $D$. Let $D'$ be any subset of $D$ such that the ratio between the diameters of any pair of elements in $D'$ is at most $\lambda$. Then $D'$ is said to be \emph{$\lambda$-related}. $D$ is said to be \emph{$(\lambda, \kappa)$-thick} if no point is covered by more than $\kappa$ elements of any $\lambda$-related subset of $D$.


The pathwidth of some graph $G$, denoted by $\pw(G)$, is the minimum integer $k\geq 1$, such that there exists a sequence $C_1,\ldots,C_{\ell}$ of subsets of $V(G)$ of cardinality at most $k+1$, such that for all $\{u,v\}\in E(G)$, there exists $i\in \{1,\ldots,\ell\}$ with $\{u,v\}\subseteq C_i$, and for all $w\in V(G)$, for all $i_1<i_2<i_3 \in \{1,\ldots,\ell\}$, if $w\in C_{i_1}\cap C_{i_3}$ then $w\in C_{i_2}$.


\subsection{Organization}
The rest of the paper is organized as follows.
In Section \ref{sec:separator} we derive a separator Theorem for a set of balls whose set of centers has bounded fractal dimension.
In Section \ref{sec:exact} we present our exact algorithms for TSP and RSMT.
In Section \ref{sec:fpt} we give a fixed-parameter algorithm for Independent Set of unit balls.
In Section \ref{sec:approx} we give approximation schemes for packing and covering unit balls.
Finally, in Section \ref{sec:spanners} we present our spanner construction.

\section{A separator theorem}\label{sec:separator}

In this section we prove a separator theorem for a set of $d$-balls intersecting a set of points with bounded fractal dimension. 
Subsequently, this result will form the basis for some of our algorithms.
The proof uses an argument due to Har-Peled \cite{har2011simple}.

\begin{theorem}\label{sep}
Let $d\geq 2$ be some integer, and let $\delta\in (0,d]$ be some real number.
Let $P \subset \mathbb{R}^d$ such that $\dimF(P)= \delta$. 
Let $B$ be a $(\lambda,\kappa)$-thick set of $d$-balls in $\mathbb{R}^d$, with $|B|=n$, $\lambda \geq 2$ and such that for all $b \in B$ we have $b \cap P \neq \emptyset$.
Then there exists a $(d-1)$-sphere $C$ such that at most $(1-2^{-O(d)})n$ of the elements in $B$ are entirely contained in the interior of $C$, at most $(1-2^{-O(d)})n$ of the elements in $B$ are entirely outside $C$,
and
\begin{align*}
|A| &= \left\{\begin{array}{ll}
O\left(\kappa (5\lambda)^d 6^{\delta} \frac{\lambda}{1-\lambda^{(1-\delta)}} n^{1-1/\delta}\right) & \text{ if } \delta>1\\
O\left(\kappa (5\lambda)^d 6^{\delta}\log n\right) & \text{ if } \delta=1\\
O\left(\frac{\kappa (5\lambda)^d 6^{\delta}}{\lambda^{1-\delta}-1}\right) & \text{ if } \delta<1
\end{array}\right.,
\end{align*}
where
$A=\{b\in B : \diam(b) \leq \diam(C) \text{ and } b\cap C\neq \emptyset\}$.
\end{theorem}

\begin{proof}
It is known that any ball in $\mathbb{R}^d$ of radius $r$ can be covered by at most $k(d) = 2^{O(d)}$ balls of radius $\frac{r}{2}$. Let $C'$ be the $d$-ball of minimum radius that contains at least $\frac{1}{k(d)+1} n$ of the elements in $B$, breaking ties by choosing the ball that contains the maximum number of elements in $B$.
Let $\origin$ denote the origin in $\mathbb{R}^d$.
Without loss of generality we can scale and translate the elements of $B$ and $P$ until the radius of $C'$ is $1$ and it is centered at $\origin$. Now, let $B^*$ denote the set of $d$-balls in $B$ of diameter less than or equal to $4$ after scaling. We pick uniformly at random $r \in [1,2]$ and let $C=\sphere(\origin,r)$. Now we are ready to obtain an upper bound on the number of elements of $B^*$ that intersect $\sphere(\origin,r)$ in expectation.

Consider any $d$-ball $b \in B^*$ of diameter $x$. The probability that $\sphere(\origin,r)$ intersects $b$ is at most $x$. Now let $ M_1=\{b\in B^* : \diam(b) \leq n^{\frac{-1}{\delta}} \text{ and } b\cap \sphere(\origin,r)\neq \emptyset\} \text{ and }
M_2=\{b\in B^* :   n^{\frac{-1}{\delta}} < \diam(b) \leq 4  \text{ and } b\cap \sphere(\origin,r)\neq \emptyset\}$. $|M_1|$ in expectation is at most $O(n^{1 - \frac{1}{\delta}})$ as $|B^*| \leq n$. It remains to bound the expected value of $|M_2|$.

Let $B_i = \{b\in B^* :   \lambda^i n^{\frac{-1}{\delta}} < \diam(b) \leq \min \{ \lambda^{i+1} n^{\frac{-1}{\delta}} , 4 \}  \text{ and } b\cap \sphere(\origin,r)\neq \emptyset\}$. Let $n_i$ denote $|B_i|$. We will construct a $\lambda^i n^{\frac{-1}{\delta}}$-net of $P$ as follows. Let $B'_i = B_i$. Let $\pi$ be some arbitrary ordering of the elements of $B'_i$. In the sequence determined by $\pi$ pick the next $d$-ball $b$ from $B'_i$. Remove all $d$-balls from $B'_i$ that are entirely within a ball of diameter $5 \cdot \lambda^{i+1}n^{\frac{-1}{\delta}}$ centered at the center of $b$. Repeat this procedure for the next element determined by $\pi$ until all the remaining $d$-balls in $B'_i$ have been visited. From the fact that $B$ is $(\lambda,\kappa)$-thick  we have that there can be at most $\kappa 5^d \lambda^d$ elements in $B'_i$ that are contained within a ball of diameter $5 \cdot \lambda^{i+1} n^{\frac{-1}{\delta}}$. This implies that we retain at least a constant fraction of the elements of $B_i$ in $B'_i$. Now from each $b \in B'_i$ pick a point $p_b$ that also belongs to $P$ and take the union of all such points to get a set of points $N_i$. From the choice of $d$-balls in the above argument $|N_i| \geq \frac{1}{\kappa 5^d \lambda^d} n_i$ and $N_i$ is $\lambda^i n^{\frac{-1}{\delta}}$-packing. We can add more points from $P$ to $N_i$ to obtain a $\lambda^i n^{\frac{-1}{\delta}}$-net $N'_i$. We have that  $|N_i| \leq |N'_i \cap \ball(\origin,6)| \leq O((\frac{6}{\lambda^i n^{\frac{-1}{\delta}}})^{\delta})$ since $\dimF(P)= \delta $ and the points of $N_i$ are contained within the ball of radius $6$ centered at the origin. This implies that $|B_i| \leq \ O( \kappa (5\lambda)^d 6^{\delta}
\lambda^{-i \delta}n)$. Since the $d$-balls in $B_i$ are intersected by $\sphere(\origin,r)$ with probability at most $\lambda^{i+1} n^{\frac{-1}{\delta}}$ we have that the expected number of elements of $B_i$ that are intersected by $\sphere(\origin,r)$ is $O( \kappa (5\lambda)^d 6^{\delta} \lambda^{i+1 -i\delta} n^{1 -  \frac{1}{\delta}})$. We thus get
\begin{align*}
\displaystyle \mathop{\mathbb{E}}[|M_2|] \leq \sum_{i = 0}^{\frac{\log{n}}{\delta} +2} |B_i| \lambda^{i+1} n^{\frac{-1}{\delta}} \leq \sum_{i = 0}^{\frac{\log{n}}{\delta} +2} O( \kappa (5\lambda)^d 6^{\delta}  \lambda^{i+1 -i\delta} n^{1 -  \frac{1}{\delta}}).
\end{align*} 
When $\delta > 1$ this implies  
\[
\mathop{\mathbb{E}}[|M_2|]\leq O( \kappa (5\lambda)^d 6^{\delta} (\frac{\lambda}{1-\lambda^{(1-\delta)}})n^{1 -  \frac{1}{\delta}}).
\]
When $\delta = 1$ we have
\[
\mathop{\mathbb{E}}[|M_2|]\leq O( \kappa (5\lambda)^d 6^{\delta} \lambda(\frac{\log{n}}{\delta} + 3)n^{1 -  \frac{1}{\delta}}) \leq O(\kappa (5\lambda)^d 6^{\delta} \log{n}).
\]
When $\delta < 1$ we have
\[
\mathop{\mathbb{E}}[|M_2|]\leq O( \kappa (5\lambda)^d 6^{\delta} \lambda(\frac{\lambda^{(\frac{\log{n}}{\delta} +3)(1-\delta)}-1}{\lambda^{(1-\delta)}-1})n^{1 -  \frac{1}{\delta}}) \leq O(\frac{\kappa (5\lambda)^d 6^{\delta}}{\lambda^{1-\delta}-1}).
\]

For any $r \in [1,2]$ we have that $A \subseteq B^*$. Thus
\begin{align*}
\displaystyle \mathop{\mathbb{E}}[|A|] = \mathop{\mathbb{E}}[|M_1|] + \mathop{\mathbb{E}}[|M_2|] \leq O(n^{1-\frac{1}{\delta}}) + \mathop{\mathbb{E}}[|M_2|],
\end{align*} 
which implies that
\begin{align*}
\mathop{\mathbb{E}}[|A|] &= \left\{\begin{array}{ll}
O( \kappa (5\lambda)^d 6^{\delta} (\frac{\lambda}{1-\lambda^{(1-\delta)}})n^{1 -  \frac{1}{\delta}}) & \text{ if } \delta>1\\
O(\kappa (5\lambda)^d 6^{\delta} \log{n}) & \text{ if } \delta=1\\
 O(\frac{\kappa (5\lambda)^d 6^{\delta}}{\lambda^{1-\delta}-1}) & \text{ if } \delta<1
\end{array}\right.
\end{align*}

Finally we need to ensure that $C$ separates a constant fraction of the elements of $B$.
The choice of $C'$ ensures that at least $\frac{1}{k(d) + 1} n = \frac{1}{2^{O(d)}} n$ of the elements in $B$ are entirely contained in the interior of $C$. This implies that at most $(1-2^{-O(d)}) n$ of the elements of $B$ are in the exterior of $C$. Since the $(d-1)$-ball of radius $2$ is covered by the union of at most $k(d)$ $(d-1)$-balls of unit radius we have that there are at most $\frac{k(d)}{k(d)+1} n = (1-2^{-O(d)}) n$ of the elements in $B$ contained in the interior of $C$.
 We note that the upper bound on $\mathop{\mathbb{E}}[|A|]$ remains unaltered for any choice of $C'$. We further remark that using a more complicated argument similar to the one used by Smith and Wormald \cite{smith1998geometric} a cube separator can be found that separates a constant fraction of $d$-balls where the constant is independent of $d$.
\end{proof}

\section{Exact algorithms}\label{sec:exact}

In this Section we give exact algorithms for TSP and RSMT problems on fractal dimension pointsets.

\subsection{TSP on fractal dimension pointsets}
We first use Theorem \ref{sep} with the following Lemmas due to Smith and Wormald \cite{smith1998geometric}  to obtain a separator for any optimal TSP solution.

\begin{lemma}[Smith and Wormald \cite{smith1998geometric}]\label{smithwortsp1}
Let $d\geq 2$ be some integer, and let $P \subset \mathbb{R}^d$. Let $W$ be the edge set of an optimal traveling salesman tour of the points of $P$. Let $B$ be the set of circumballs of the edges of $W$. Then $B$ is $(2,\kappa)$-thick where $\kappa = 2^{O(d)}$.
\end{lemma}

\begin{lemma}[Smith and Wormald \cite{smith1998geometric}]\label{smithwortsp2}
Let $d\geq 2$ be some integer, and let $P \subset \mathbb{R}^d$. Let $W$ be the edge set of an optimal traveling salesman tour of the points of $P$. For any $x \in \mathbb{R}^d$ let $W_x = \{ w \in W : \diam(w) \geq 1 \text{ and } w\cap \ball(x,1)\neq \emptyset \}$. Then $|W_x| \leq 2^{O(d)}$ for all $x \in \mathbb{R}^d$.
\end{lemma}

\begin{theorem}\label{septsp}
Let $d\geq 2$ be some integer, and let $\delta\in (0,d]$ be some real number. Let $P$ be a set of $n$ points in $\mathbb{R}^d$ with $\dimF(P)= \delta$. Let $W$ be the set of edges of any optimal Euclidean TSP tour of $P$. Then there exists a $(d-1)$-sphere  $C$ such that at most $(1-2^{-O(d)})n$ points in $P$ are contained in the interior of $C$, at most $(1-2^{-O(d)})n$ points in $P$ are contained outside $C$, and 
\begin{align*}
|W_C| &= \left\{\begin{array}{ll}
O( n^{1-1/\delta}) & \text{ if } \delta>1\\
O(\log n) & \text{ if } \delta=1\\
O\left(\frac{1}{2^{1-\delta}-1}\right) & \text{ if } \delta<1
\end{array}\right.,
\end{align*}
where
 $W_C=\{w\in W : w\cap C\neq \emptyset\}$. 
\end{theorem}

\begin{proof}
Let $B$ be the set of circumballs of the edges in $W$. From  Lemma \ref{smithwortsp1}  we have that $B$ is $(2,2^{O(d)})$-thick. Every ball in $B$ contains an edge in $W$ and therefore also two points in $P$. Therefore we can use Theorem \ref{sep} on $B$ to find a separator $C$.
It remains to bound the number of edges in $W$ that are intersected by $C$. Let $W_1 = \{ w \in W : \diam(w) \leq \diam(C) \text{ and } w\cap C\neq \emptyset  \}$ and $W_2 = \{ w \in W : \diam(w) > \diam(C) \text{ and } w\cap C\neq \emptyset \}$. Therefore $W_C = W_1 \cup W_2$. Let $B_1$ denote the circumballs of the edges in $W_1$ and $B_2$ denote the circumballs of the edges in $W_2$. If an edge in $W_1$ is intersected by $C$ then the corresponding circumball in $B_1$ is also intersected by $C$. From Theorem \ref{sep} we have that
\begin{align*}
|W_1| &= \left\{\begin{array}{ll}
O( n^{1-1/\delta}) & \text{ if } \delta>1\\
O(\log n) & \text{ if } \delta=1\\
O\left(\frac{1}{2^{1-\delta}-1}\right) & \text{ if } \delta<1
\end{array}\right.
\end{align*}
W.l.o.g.~we can assume that $C$ has unit radius and is centered at the origin by scaling and translation. Therefore any edge in $W_2$ also intersects the unit ball centered at the origin. Combining this with Lemma \ref{smithwortsp2} we have that $|W_2| \leq O(1)$.
Since $|W_C|\leq |W_1|+|W_2|$, this concludes the proof.
\end{proof}

We now use Theorem \ref{septsp} to obtain an exact algorithm for TSP. 
We note that the $O$-notation hides a factor of $n^{O(1)^d}$.

\begin{theorem}\label{tsp}
Let $d\geq 2$ be some fixed integer, and let $\delta\in (0,d]$ be some real number.
Let $P$ be a set of $n$ points in $\mathbb{R}^d$ with $\dimF(P)= \delta $. Then for any fixed $d$ an optimal Euclidean TSP tour for $P$ can be found in time $T(n)$, where
\begin{align*}
T(n) &= \left\{\begin{array}{ll}
n^{O(n^{1-1/\delta})} & \text{ if } \delta>1\\
n^{O(\log^2 n)} & \text{ if } \delta=1\\
n^{O(\log n)} & \text{ if } \delta<1
\end{array}\right.
\end{align*}
\end{theorem}

\begin{proof}
First we observe that the $(d-1)$-sphere separator $C$ described in Theorem \ref{septsp} can be assumed to intersect at least $d+1$ points in $P$. This is because we can always decrease the radius of $C$ without changing $W_C$ until at least one point in $P$ lies on it. 
 We exhaustively consider all separating $(d-1)$-spheres to find the separator from Theorem \ref{septsp}. Since every relevant $(d-1)$-sphere is uniquely defined by at most $d+1$ points of $P$ intersecting it, there are at most $n^{O(d)}$ spheres to consider.
Let $f(n,\delta)$ denote the number of edges intersected by the separator $C$. From Theorem \ref{septsp} we have that
\begin{align*}
f(n,\delta) &= \left\{\begin{array}{ll}
O(n^{1-1/\delta}) & \text{ if } \delta>1\\
O(\log n) & \text{ if } \delta=1\\
O(1) & \text{ if } \delta<1
\end{array}\right.
\end{align*} 
We guess a set $E'$ of at most $f(n,\delta)$ edges in the optimal tour that intersect $C$.
For each such guess $E'$, we also guess the permutation of $E'$ defined by the order in which the optimal tour traverses the edges in $E'$.
For each such permutation we solve the two sub-problems in the exterior and interior of the separator respecting the boundary conditions. The resulting running time is $T(n) \leq n^{O(d)} n^{O(f(n,\delta))} 2 T\left((1-2^{-O(d)})n\right)$ which implies that for any fixed $d$ implies the assertion.
\end{proof}

\subsection{Finding the rectilinear Steiner minimal tree in $\mathbb{R}^2$}
Given $P \subset \mathbb{R}^d $ a set of $n$ points. A
\emph{Rectilinear Steiner Tree} (RST) is a geometric graph connecting all the points in $P$ and consisting only of line segments parallel to the coordinate axes. The length of a RST is the sum of the lengths of the line segments in it. A \emph{Rectilinear Steiner Minimal Tree} (RSMT) is a RST of minimal length.

We will use the following lemmas to prove our theorem.
\begin{lemma}\label{rsmtlem}
Let $d \geq 2$ be some integer. Let $P \subset \mathbb{R}^d$. Let $S$ be an RSMT of $P$. Let $B$ be the set of circumballs of the line segments of $S$. Then  for any $\lambda > 0$, $B$ is $(\lambda,O(\lambda^d))$-thick.
\end{lemma}

\begin{proof}
Consider any edge $xy \in S$. Let $m_{xy}$ denote the mid-point of $xy$.
 Let the diamond of $xy$ denoted by $D_{xy}$ be defined as $D_{xy} = \{p : p \in \mathbb{R}^d \text{ and } \|m_{xy}-p\|_1 \leq \frac{1}{2} \|m_{xy}-x\|_1 \}$. The $d$-volume of $D_{xy}$ is $O( \|x-y\|^d)$. Smith and Wormald\cite{smith1998geometric} proved that for any pair of edges $xy,wz \in S$, $D_{xy}$ and $D_{wz}$ (called diamonds) are disjoint. Let $B' \subseteq B$ be $\lambda$-related. Let the minimal diameter of any ball in $B'$ be $\alpha$. Let $qr \in S$ be an edge of minimal length $\alpha$ whose circumball is in $B'$. This implies that the maximal diameter of any ball in $B'$ is at most $\lambda \alpha$. Now consider any point $p \in \mathbb{R}^d$. Any element of $B'$ that covers $p$ lies within the $\ball(p,\lambda \alpha)$. Since the diamonds of any pair of edges in $S$ are disjoint it follows that the number of circumballs covering $p$ is at most $\frac{\volume(\ball(p,\lambda \alpha))}{\volume(D_{qr})} = O(\frac{(\lambda \alpha)^d}{\alpha^d}) = O(\lambda^d)$.
\end{proof}

\begin{lemma}\label{rsmtlem2}
Let $d \geq 2$ be some integer. Let $P \subset \mathbb{R}^d$. Let $S$ be an RSMT of $P$. Let $B$ be the set of circumballs of the line segments of $S$. For any $x \in \mathbb{R}^d$ let $S_x = { s \in S : \diam(s) \geq 1 \text{ and } s\cap \ball(x,1)\neq \emptyset }$. Then $|S_x| \leq O(2^d)$ for all $x \in \mathbb{R}^d$.
\end{lemma}

\begin{proof}
Consider any line segment $t \in S_x$. We have that the diamond of $t$ $D_t$ occupies at least $O(1^d)$ $d$-volume in $\ball(x,2)$. Since the diamonds of all of these edges are disjoint we have that $|S_x| \leq \frac{\volume(\ball(x,2))}{O(1^d)} = O(2^d)$.
\end{proof}

\begin{lemma}[Smith and Wormald \cite{smith1998geometric}]\label{smithworrsmt}
Let $d>0$ be some integer. Let $P \subset \mathbb{R}^d$. Let $S$ be an RSMT of $P$. There are only $n^d$ possible locations for Steiner points in $S$. Consequently there are at most $n^{O(d)}$ possible line segments that can be part of $S$.
\end{lemma}


\begin{theorem}\label{seprsmt}
Let $\delta \in (0,2]$ be some real number. Let $P \subset \mathbb{R}^2$ such that $\dimF(P)= \delta $. Let $S$ be the set of line segments of an RSMT of $P$. Then there exists a $1$-sphere $C$ such that at least $\frac{1}{8}$ of the points in $P$ are contained in the interior of $C$ and at least $\frac{1}{8}$ of the points in $P$ are contained outside $C$. Let $S_C=\{s\in S : s\cap C\neq \emptyset\}$. Then we have,
\begin{align*}
|S_C| &= \left\{\begin{array}{ll}
O( n^{1-1/\delta}) & \text{ if } \delta>1\\
O(\log n) & \text{ if } \delta=1\\
O(\frac{1}{\lambda^{1-\delta}-1}) & \text{ if } \delta<1
\end{array}\right.
\end{align*}
\end{theorem}

\begin{proof}
Let $B$ be the set of circumballs of the line segments in $S$. From  lemma \ref{rsmtlem} we have that $B$ is $(2,O(2^2))$-thick. W.l.o.g we may assume that every line segment in $S$ has at most one Steiner point as an end point since $P \subset \mathbb{R}^2$. Every circumball in $B$ contains a line segment in $S$ and therefore at least one point in $P$. Therefore we can use theorem \ref{sep} on $B$ to find a separator $C$. We choose $C'$ to be the $2$-ball of minimum radius that contains $\frac{1}{8}$ of the points in $P$. This gives us a separator $C$ that separates a constant fraction of the points in $P$.

Now it remains to bound the number of edges in $W$ that are intersected by $C$. Let $S_1 = \{ s \in S : \diam(s) \leq diam(C) \text{ and } s\cap C\neq \emptyset  \}$ and $S_2 = \{ s \in S : \diam(s) > diam(C) \text{ and } s\cap C\neq \emptyset \}$. Therefore $S_C = S_1 \cup S_2$. Let $B_1$ denote the circumballs of the line segments in $S_1$ and $B_2$ denote the circumballs of the line segments in $S_2$. If a line segment in $S_1$ is intersected by $C$ then the corresponding circumball in $B_1$ is also intersected by $C$. From theorem \ref{sep} we have that, 
\begin{align*}
|S_1| &= \left\{\begin{array}{ll}
O( n^{1-1/\delta}) & \text{ if } \delta>1\\
O(\log n) & \text{ if } \delta=1\\
O(\frac{1}{\lambda^{1-\delta}-1}) & \text{ if } \delta<1
\end{array}\right.
\end{align*}
Without loss of generality we can assume that $C$ has unit radius and is centered at the origin due to scaling and translation. Therefore any line segment in $S_2$ also intersects the unit ball centered at the origin. Combining this with lemma \ref{rsmtlem2} we have that $|S_2| \leq O(1)$. This implies that, 
\begin{align*}
|S_C| &= \left\{\begin{array}{ll}
O( n^{1-1/\delta}) & \text{ if } \delta>1\\
O(\log n) & \text{ if } \delta=1\\
O(\frac{1}{\lambda^{1-\delta}-1}) & \text{ if } \delta<1
\end{array}\right.
\end{align*}

\end{proof}

\begin{theorem}
Let $\delta \in (0,2]$ be some real number. Let $P \subset \mathbb{R}^2$ such that $\dimF(P)= \delta $. Then an RSMT of $P$ can be found in running time $T(n)$, where 
\begin{align*}
T(n) &= \left\{\begin{array}{ll}
2^{O(n^{1-1/\delta} \log n)} & \text{ if } \delta>1\\
n^{O(\log^2 n)} & \text{ if } \delta=1\\
n^{O(\log n)} & \text{ if } \delta<1
\end{array}\right.
\end{align*}
\end{theorem}

\begin{proof}
Let $S$ be an RSMT of $P$. Let 
\begin{align*}
f(n,\delta) &= \left\{\begin{array}{ll}
2^{O(n^{1-1/\delta} \log n)} & \text{ if } \delta>1\\
n^{O(\log n)} & \text{ if } \delta=1\\
n^{O(1)} & \text{ if } \delta<1
\end{array}\right.
\end{align*}
Theorem \ref{seprsmt} implies that there exists a separator $C$ intersecting at most $f(n,\delta)$ elements of $S$ and separating a constant fraction of the points in $P$ for a certain choice of $C'$. Like in theorem \ref{tsp} we first choose $C'$ by finding the smallest $2$-ball containing at least $\frac{1}{8}$ points in $P$. We can do this exhaustively in time $n^{O(1)}$ because every relevant $2$-ball is uniquely determined by at most $3$ points in $P$. Next we fix the center of our separator to be the center of $C'$ and exhaustively consider all relevant radii. Since the separator $C$ can be assumed to intersect at least one point in $P \cup S$ where $S$ is the set of possible Steiner points we only need to consider at most $|P \cup S|$ different radii. From lemma \ref{smithworrsmt} we have that this can be done in time at most $n^{O(2)}$. From theorem \ref{seprsmt} we can also guess the line segments intersected by $C$ in time $n^{O(2f(n,\delta))}$. Then we can guess the boundary condition. Given $M$ crossing line segments there are at most $O(n^M)$ possible boundary conditions. Here $M \leq f(n,\delta)$ so we can guess these in time at most $O(n^{f(n,\delta)})$. Finally we can solve the two smaller subproblems in the interior and exterior of $C$. The running time follows the recursion $T(n) \leq n^{O(1)} \cdot n^{O(2f(n,\delta))} \cdot O(n^{f(n,\delta)}) \cdot [2T(\frac{7n}{8})]$. This implies that,
\begin{align*}
T(n) &= \left\{\begin{array}{ll}
2^{O(n^{1-1/\delta} \log n)} & \text{ if } \delta>1\\
n^{O(\log^2 n)} & \text{ if } \delta=1\\
n^{O(\log n)} & \text{ if } \delta<1
\end{array}\right.
\end{align*}
\end{proof}


\section{Parameterized problems}\label{sec:fpt}

In this section we present an algorithm for the parameterized version of the Independent Set problem on a set of unit $d$-balls in $\mathbb{R}^d$, where set of centers of the $d$-balls has bounded fractal dimension. We first prove a separator theorem which will be used in the algorithm.

\begin{theorem}\label{septhmparam}
Let $d \geq 2$ be an integer. Let $\delta \in (0,d]$ be a real number. Let $P$ be a set of $n$ points in $\mathbb{R}^d$ with $\dimF(P) = \delta$. Let $D = \{ \ball(x,1): x\in P\}$.
 Let $D' \subseteq D$ be a set of disjoint elements of $D$ such that $|D'| = k$. Then there exists $c \in \mathbb{R}^d$ and $r>0$ such that at most $H$ $d$-balls in $D'$ intersect $\sphere(c,r)$ and at 
 most $(1-2^{-O(d)})k$ $d$-balls in $D'$ are contained on either side (interior and exterior) of $\sphere(c,r)$ where
\begin{align*}
H &= \left\{\begin{array}{ll}
O(k^{1-\frac{1}{\delta}}) & \text{ if } \delta>1\\
O(1) & \text{ if } \delta \leq 1
\end{array}\right.
\end{align*} 
\end{theorem}

\begin{proof}
Let $P'$ denote the set of centers of the $d$-balls in $D'$.
We have $|P'| = |D'| = k$. Also since the $d$-balls in $D'$ are disjoint we have that $P'$ is a $2$-packing of $P$.
Consider any $c \in \mathbb{R}^d$ and any $r \geq 1$.
Consider a random $(d-1)$-sphere $\sphere(c,r')$ with radius $r' \in [r,2r]$ chosen uniformly at random. Now we can bound the number of $d$-balls in $D'$ that intersect the $\sphere(c,r')$. First we note that the center of any $d$-ball in $D'$ that potentially intersects $\sphere(c,r')$ lies within $\ball(c,2r+1)$. Therefore the number of $d$-balls that potentially intersect $\sphere(c,r')$ is at most $|P' \cap \ball(c,2r+1)|$. Since $P'$ is a $2$-packing that can be augmented into a $2$-net by only adding points, we have that $|P' \cap \ball(c,2r+1)| \leq O((\frac{2r+1}{2})^\delta) = O(r^\delta)$. Therefore we have that the number of $d$-balls that potentially intersect $\sphere(c,r')$ is at most $\min \{k,O(r^\delta)\}$. Any $d$-ball in $D'$ intersects $\sphere(c,r')$ with probability at most $\frac{2}{r}$ since $r'$ is chosen uniformly at random from the interval $[r,2r]$. So in expectation the number of $d$-balls in $D'$ that intersect $\sphere(c,r')$ is at most $\min \{k \cdot \frac{2}{r} ,O(r^\delta) \cdot \frac{2}{r} \} $. When $r \leq k^{\frac{1}{\delta}}$ and $\delta > 1$ this is at most $O(r^\delta) \cdot \frac{2}{r} = O(k^{1-\frac{1}{\delta}})$, when $r \leq k^{\frac{1}{\delta}}$ and $\delta \leq 1$ this is at most $O(r^\delta) \cdot \frac{2}{r} = O(r^{\delta -1}) = O(1)$, and when $r > k^{\frac{1}{\delta}}$ this is again at most $k \cdot \frac{2}{r} = O(k^{1-\frac{1}{\delta}})$. This implies that there exists some specific  $r' \in [r,2r]$ such that the number of $d$-balls in $D'$ that intersect $\sphere(c,r')$ is at most $O(k^{1-\frac{1}{\delta}})$ when $\delta > 1$ and $O(1)$ when $\delta \leq 1$.
 
Now it remains to specify our choice of $c$ and $r$ so that $\sphere(c,r')$ induces a balanced separator. We will use the fact that for any $r>0$ any $d$-ball of radius $2r$ can be covered by at most $g(d) = 2^{O(d)}$ $d$-balls of radius $r$. Let $c$ and $r$ be chosen such that $\sphere(c,r)$ is the $(d-1)$-sphere with minimum radius that also contains in its interior $\frac{1}{g(d)+1}k$ elements of $D'$. Since the $d$-balls have unit radius it follows that the $r \geq 1$. This ensures that there are at least $\frac{1}{2^{O(d)}}$ elements of $D'$ in the interior of $\sphere(c,r)$ and therefore at most $(1-2^{-O(d)})k$ elements of $D'$ in the exterior of $\sphere(c,r)$. We have that $\ball(c,r')$ is contained within $\ball(c,2r)$. Since $\ball(c,2r)$ can be covered by at most $g(d)$ $d$-balls of radius $r$, by our choice of $c$ and $r$ we have that $\sphere(c,r')$ encloses at most $\frac{g(d)}{g(d)+1}k = (1-2^{-O(d)})k$ $d$-balls in $D'$ concluding the proof.
\end{proof}

\begin{theorem}
Let $d \geq 2$ be an integer. Let $\delta \in (0,d]$ be a real number. Let $P$ be a set of $n$ points in $\mathbb{R}^d$ with $\dimF(P) = \delta$. Let $D = \{ \ball(x,1): x\in P\}$.
Then there exists an algorithm that computes an independent set in $D$ of size $k$, if one exists, in time $T(n,k)$, where for any fixed $d$ we have
\begin{align*}
T(n,k) &= \left\{\begin{array}{ll}
n^{O(k^{1-1/\delta})} & \text{ if } \delta>1\\
n^{O(\log k)} & \text{ if } \delta \leq 1
\end{array}\right.
\end{align*} 
 
\end{theorem}

\begin{proof}
Let $D' \subseteq D$ denote the set of $k$ disjoint $d$-balls in any fixed optimal solution. Let $P'$ denote the set of centers of the $d$-balls in $D'$. We have $|P'| = |D'| = k$. We use a divide and conquer approach using the separator from Theorem \ref{septhmparam}. First we guess the center $c$ and radius $r$ of the smallest $(d-1)$-sphere enclosing $\frac{1}{g(d)+1}$ of the $d$-balls in $D'$. W.l.o.g.~we can assume that there exist a set of $d$-balls in $D'$ that are tangential to $\sphere(c,r)$ and are enclosed by $\sphere(c,r)$, of cardinality $d+1$. Moreover $\sphere(c,r)$ is uniquely defined by the $d$-balls that it is tangential to. This implies that $\sphere(c,r-1)$ intersects at least $d+1$ points in $P$ and can be uniquely defined by at most $d+1$ points in $P$. We can exhaustively guess $c$ by searching through all $(d-1)$-spheres uniquely defined by at most $d+1$ points in $P$ in time $n^{O(d)}$. Next we can assume w.l.o.g.~that $\sphere(c,r')$ from Theorem \ref{septhmparam} is tangential to at least one $d$-ball in $D'$ (otherwise $r'$ can be increased or decreased until this condition is met without altering the set of $d$-balls in $D'$ that are intersected by the separator). This means that given a fixed center $c$ we need to search through at most $2n$ different radii to guess $r'$. We enumerate over all such separators. For each such separator we again enumerate over all ways to pick the $d$-balls in $D'$ that are intersected. This can be done in time $n^{O(k^{1-1/\delta})}$ when $\delta > 1$ and $n^{O(1)}$ when $\delta \leq 1$. Therefore we have  $T(n,k) = n^{O(d)} \cdot O(n) \cdot n^{O(k^{1-1/\delta})} \cdot 2\cdot T\left(n,(1-2^{-O(d)})k\right)$ when $\delta > 1$ or $T(n,k) = n^{O(d)} \cdot O(n) \cdot n^{O(1)} \cdot 2\cdot T\left(n,(1-2^{-O(d)})k\right)$ when $\delta \leq 1$, which solves to the desired bound.
\end{proof}

\section{Approximation schemes}\label{sec:approx}

In this section we describe polynomial time approximation schemes for covering and packing problems.
We use the approach of Hochbaum and Maass \cite{hochbaum1985approximation}.

\begin{theorem}\label{coveringscheme}
Let $d\geq 2$ be some integer, and let $\delta\in (0,d]$ be some real number. Let $P$ be a set of $n$ points in $\mathbb{R}^d$ with $\dimF(P) = \delta$. Then there exists a polynomial time approximation scheme which given a natural number $l > 0$ and any $\eps >0$, computes a $(1 + \frac{d}{l})$-approximation to the $\eps$-cover of $P$,
in time $l^{d+\delta} n^{O({(l \sqrt{d})^{\delta})}}$.
\end{theorem}

\begin{proof}
Let $A$ be a $d$-rectangle that encloses the points in $P$. Consider a set of hyperplanes perpendicular to an axis of the ambient space that subdivide $A$ into strips of width $2l\eps$, which are left closed and right open.
This gives a partition $P_0$ of $A$ where each strip has width $2l \eps$. Now for any integer $i$ where $0 < i < l$ we shift the hyperplanes that define the partition $P_0$ by $2i \eps$ to the right to get the partition $P_i$. Let $S = \{P_0,P_1,\ldots,P_{l-1}\}$. Let $\OPT$ be the optimal $\eps$-cover of $P$. Let $D$ be the set of $d$-balls of radius $\eps$ centered at the points in $\OPT$. Any $d$-ball in $D$ intersects the hyperplanes from at most one partition in $S$. Therefore there exists a partition $P_i$ such that at most $\frac{|D|}{l}$ $d$-balls in $D$ are intersected by the hyperplanes defining $P_i$. In other words at most $\frac{|D|}{l}$ $d$-balls in $D$ intersect more than one strip in $P_i$. Now we consider partitioning $A$ similarly along each axis  to get a grid of hypercubes of side length $2l \eps$, which we call \emph{cells}. Using the argument described above it follows that there exits a partition $P'$ such that at most $\frac{d|D|}{l}$ $d$-balls in $D$ intersect more than one cell in $P'$.

Now consider a cell $C$ of side length $2l \eps$. Since $\dimF(P) = \delta$ and $C$ is contained in a ball of radius $\sqrt{d}l \eps$ we have that there exists an $\eps$-cover of the points in $C$ of cardinality at most $O(\frac{\sqrt{d}l \eps}{\eps})^{\delta} = O(\sqrt{d}l )^{\delta}$. 

We combine the above observations to obtain our algorithm as follows. The algorithm enumerates all $l^d$ partitions of $P$ into cells of side length $2 l \eps$. Next it enumerates exhaustively all $\eps$-covers of cardinality at most $O((\sqrt{d}l)^{\delta})$ for each cell. Since verifying whether a set of points is a valid cover takes time $O(n(\sqrt{d}l)^\delta)= O(nl^{\delta})$ this step overall takes time at most $n^{O((\sqrt{d}l)^{\delta})} \cdot l^{\delta}$. Finally the algorithm takes the union of the $\eps$-covers of all the cells to get an $\eps$-cover of $P$ and returns the best solution over all partitions. Since there exists at least one partition where  at most $\frac{d|D|}{l}$ $d$-balls in $D$ intersect more than one cell in the partition, we have that the size of the solution returned is at most $(1+\frac{d}{l})|D| = (1+\frac{d}{l})|\OPT|$. The running time of the algorithm is $l^d \cdot n^{O((\sqrt{d} l)^{\delta})} \cdot l^{\delta} = l^{d+\delta} n^{O((l \sqrt{d})^{\delta})}$.
\end{proof}


\begin{theorem}
Let $d\geq 2$ be some integer, and let $\delta\in (0,d]$ be some real number. Let $P$ be a set of $n$ points in $\mathbb{R}^d$ with $\dimF(P) = \delta$. There exists a polynomial time approximation scheme which given a natural number $l > 0$ and any $\eps >0$, computes a $(1 + \frac{d}{l-d})$-approximation to the $\eps$-packing of $P$, in time $l^{d+\delta} n^{O({(l \sqrt{d})^{\delta})}}$.
\end{theorem}

\begin{proof}
We use the partitioning approach described in Theorem \ref{coveringscheme}. We consider cells of side length $l\eps$. Since any $\eps$-packing can be augmented into an $\eps$-net we have that any $\eps$-packing of the points in a cell has cardinality at most $O((\frac{\sqrt{d}l}{2})^{\delta})$. We consider $\eps$-packings for each cell where the points in the packing are all at least distance $\frac{\eps}{2}$ from the boundary of the cell; this ensures that the $d$-balls of radius $\frac{\eps}{2}$ centered at these points do not intersect multiple cells. Then we take the union of these points over all cells and take the minimum cardinality set over all partitions. The running time is $l^{d+\delta} n^{O({(l \sqrt{d})^{\delta})}}$ by the same reasoning used in Theorem \ref{coveringscheme}. Let $\OPT$ be the optimal $\eps$-packing of $P$. Since at most $\frac{d}{l}|\OPT|$ $d$-balls in the optimal packing intersect more than one cell we have that the solution returned by the algorithm has cardinality at least $(1 - \frac{d}{l})|\OPT|$, as required.
\end{proof}


\section{Spanners and pathwidth}\label{sec:spanners}


We remark that several other constructions of $(1+\eps)$-spanners for finite subsets of $d$-dimensional Euclidean space are known.
However, they do not yield graphs of small pathwidth.
Here we use a construction that is a modified version of the spanner due to Vaidya \cite{vaidya1991sparse}.
Let $P$ be a set of $n$ points in $\mathbb{R}^d$.
Let us first recall the construction from \cite{vaidya1991sparse}.
Let $\eps>0$.
We will define a graph $G$ with $V(G)=P$, that is a $(1+\eps)$-spanner for $P$.

Let $I_1,\ldots,I_d\subset \mathbb{R}$ be intervals, all having the same length, and such that each $I_i$ is either closed, open, or half-open.
Then we say that $b=I_1\times \ldots\times I_d$ is a \emph{box}.
We define $\size(b)$ to be the length of the interval $I_1$.
For each $i\in \{1,\ldots,d\}$, let $\psi(b)_i$ be the center of $I_i$, and define the half-spaces
$L_i(b)=\{(x_1,\ldots,x_d)\in \mathbb{R}^d : x_i<\psi_i(b)\}$
and
$R_i(b)=\{(x_1,\ldots,x_d)\in \mathbb{R}^d : x_i\geq\psi_i(b)\}$
Let $\ImSucc(b)$ be the set of boxes such that 
$\ImSucc(b) = \left\{b' : b'=b\cap \left(\bigcap_{i=1}^d f_i\right), \text{ where for all } i\in \{1,\ldots,d\}, f_i=L_i(b) \text{ or } f_i = R_i(b) \right\}$.
We also define $\shrunk(b)$ to be some box satisfying the following conditions:

(1)
If $|b\cap P|\leq 1$ then $\shrunk(b)=b\cap P$. Note that we allow $\shrunk(b)$ to be empty.

(2)
If $|b\cap P|\geq 2$ then $\shrunk(b)$ is some minimal box contained in $b$ with $\shrunk(b)\cap P=b\cap P$.
Note that if there are multiple choices for $\shrunk(b)$, then we choose one arbitrarily.

For some box $b$ with $|b\cap P|\geq 2$, we define
$\Succ(b)$ to be the set of boxes such that
$\Succ(b) = \left\{b' : \text{ there exists } b''\in \ImSucc(b) \text{ s.t.~} b''\cap P\neq \emptyset \text{ and } b'=\shrunk(b'') \right\}$.
If $|b\cap P|\leq 1$, then we define $\Succ(b)=\emptyset$.

The \emph{box-tree} of $P$ is defined to be a tree $T$ where every node is some box.
We set the root of $T$ to be some minimal box $b^*$ containing $P$.
For each $b\in V(T)$, the set of children of $b$ in $T$ is $\Succ(b)$. Note that $|b\cap P| = 1$ if and only if $b$ is a leaf of $T$.
For each $b\in V(T)\setminus \{b^*\}$ we denote by $\father(b)$ the father of $b$ in $T$.

For each $b\in V(T)$ let
\begin{align*}
\Near(b) &= \left\{ b'\in V(T)\setminus\{b^*\} : \size(b')<\size(b)\leq \size(\father(b')) \text{ and } \dist(b,b') \leq \frac{6\sqrt{d}}{\eps}  \size(b) \right\}.
\end{align*}

It follows by the construction that for each $b\in V(T)$, we have $b\cap P\neq \emptyset$.
For each $b\in V(T)$ pick some arbitrary point $\rep(b)\in b\cap P$.
We say that $\rep(b)$ is the \emph{representative} of $b$.
We further impose the constraint that for each non-leaf $b\in V(T)$, if $b'$ is the unique child of $b$ with $\rep(b)\in b'$, then $\rep(b')=\rep(b)$.
This implies that for every $b\in V(T)$, there exists a branch in $T$ starting at $b$ and terminating at some leaf, such that all the boxes in the branch have the same representative as $b$.
We remark that this additional requirement is not necessary in the original construction of Vaidya \cite{vaidya1991sparse}.

We define $E(G)=E_1\cup E_2$, 
where
$E_1 = \{\{\rep(b),\rep(b')\} : b\in V(T), b'\in \Succ(b), \rep(b)\neq \rep(b')\}$
and
$E_2 = \{\{\rep(b),\rep(b')\} : b\in V(T), b'\in \Near(\father(b))\}$.
This completes the description of the spanner construction due to Vaidya \cite{vaidya1991sparse}.
His result is summarized in the following.

\begin{theorem}[Vaidya \cite{vaidya1991sparse}]\label{thm:Vai}
$G$ is a $(1+\eps)$-spanner for $P$.
Moreover $|E(G)|= O(\eps^{-d} n)$.
\end{theorem}

For each $e=\{u,v\}\in E(G)$, let $D_e$ be the circumscribed ball for the segment $u$-$v$.
Let ${\cal D}=\bigcup_{e\in E(G)}\{D_e\}$.
For each $i\in\{1,2\}$ let ${\cal D}_i=\bigcup_{e\in E_i}\{D_e\}$.

\begin{lemma}\label{lem:D1}
${\cal D}_1$ is $(2,d^{O(d)})$-thick.
\end{lemma}

\begin{proof}
Let $r>0$ and define
$E_{1,r} = \{\{x,y\}\in E_1 : r\leq \|x-y\|_2 <2r\}$.
Let 
${\cal D}_{1,r} = \{D_e \in {\cal D}_1 : e\in E_{1,r}\}$.
It suffices to show that ${\cal D}_{1,r}$ is $d^{O(d)}$-thick.

For each $e = \{x,y\} \in E_{1,r}$ we define some unordered pair of boxes $\gamma(e)=\{B(e), B'(e)\}$, as follows.
By the definition of $E_1$, there exists some $b\in V(T)$, $b'\in \Succ(b)$, with $\rep(b)\neq \rep(b')$, such that $\{x,y\}=\{\rep(b), \rep(b')\}$.
Assume w.l.o.g.~that $x=\rep(b)$ and $y=\rep(b')$.
By the choice of the representatives, there exist some branch $b_0,\ldots,b_t$ of $T$, for some $t\geq 1$ with $b_0=b$, that terminates at some leaf $b_t$, such that $x=\rep(b)=\rep(b_0)=\ldots=\rep(b_t)$.
Since $x,y\in b$, it follows that
$r\leq \|x-y\|_2 \leq \sqrt{d} \cdot \size(b)$.
Since $b_t$ is a leaf, we have $\size(b_t)=0$.
Let $t^*\in \{1,\ldots,t\}$ be the maximum integer such that 
$\size(b_{t^*-1}) \geq r/\sqrt{d}$.
Let $A\in \ImSucc(b_{t^*-1})$ such that $b_{t^*}\subseteq A$.
Note that $\size(A)\geq r/(2\sqrt{d})$, and $\size(b_{t^*}) < r/\sqrt{d}$.
Pick some box $B(e)$, with $b_{t^*}\subseteq B(e)\subseteq A$, such that
\begin{align}
\size(B(e)) \in \left[r/(2\sqrt{d}), r/\sqrt{d}\right] \label{eq:Be_size}
\end{align}
in a consistent fashion (i.e. for a fixed choice of $b_{t^*}$ and $A$ we always pick the same box).
Similarly, let $b_0',\ldots,b_s'$ be a sequence of boxes such that $b_0'\in \Succ(b)$, with $b'\subseteq b_0'$, and $b_1',\ldots,b_s'$ is a branch of $T$ starting at $b_1'=b'$ and terminating at some leaf $b_s'$.
Arguing as before, let $s^*\in \{1,\ldots,s\}$ be the maximum integer such that 
$\size(b_{s^*-1}') \geq r/(2\sqrt{d})$.
If $s^*=1$ then let $A'\in \Succ(b'_{s^*-1})$, with $b'_{s^*}\subseteq A'$; pick some box $B'(e)$, with $b_{s^*}\subseteq B'(e)\subseteq A'$, such that 
\begin{align}
\size(B'(e)) \in \left[r/(4\sqrt{d}), r/(2\sqrt{d})\right] \label{eq:BBe_size}
\end{align}
in a consistent fashion.

We say that $e$ is \emph{charged} to $\gamma(e)$.
By construction, there exists at most one edge in $E_{1,r}$ that is charged to each pair of boxes.

By \eqref{eq:Be_size} and \eqref{eq:BBe_size} we have that for each $e\in E_{1,r}$, the pair $\gamma(e)$ consists of two boxes, each of size $\Theta(r/\sqrt{d})$. Moreover by construction and our choice of boxes we have that for any $e,f \in E_{1,r}$ $B(e)$ and $B(f)$ are disjoint or equal. Similarly $B'(e)$ and $B'(f)$ are also disjoint or equal. Thus, each point in $\mathbb{R}^d$ can be contained in at most $O(1)$ boxes in all the pairs $\gamma(e)$, for all $e\in E_{1,r}$.
Moreover, $\dist(B(e),B'(e))\leq \|x-y\| < 2r$.
Thus, each box participates in at most $(\sqrt{d})^{O(d)} = d^{O(d)}$ pairs.
For each $e\in E_{1,r}$, let $A(e)=N(B(e),r)\cup N(B'(e),r)$, where $N(X,r)$ denotes the $r$-neighborhood of $X$ in $\mathbb{R}^d$.
It follows that $\{A(e)\}_{e\in E_{1,r}}$ is $d^{O(d)}$-thick.
Since for each $e\in E_{1,r}$, we have $D_e\in A(e)$, it follows that ${\cal D}_1$ is $d^{O(d)}$-thick, as required.
\end{proof}

\begin{lemma}\label{lem:D2}
${\cal D}_2$ is $(2,(d/\eps)^{O(d)})$-thick.
\end{lemma}

\begin{proof}
Let $r>0$ and define
$E_{2,r} = \{\{x,y\}\in E_2 : r\leq \|x-y\|_2 <2r\}$.
Let 
${\cal D}_{2,r} = \{D_e \in {\cal D}_2 : e\in E_{2,r}\}$.
It suffices to show that ${\cal D}_{2,r}$ is $d^{O(d)}$-thick.

As in the proof of Lemma \ref{lem:D1}, for each $e\in {\cal D}_{1,r}$ we define some unordered pair of boxes $\gamma(e)=\{B(e), B'(e)\}$.
By the definition of $E_2$, there exists some $b\in V(T)$, $b'\in \Near(\father(b))$, such that $\{x,y\}=\{\rep(b), \rep(b')\}$.
Assume w.l.o.g.~that $x=\rep(b)$ and $y=\rep(b')$.
Thus we have
$\size(b')<\size(\father(b)) \leq \size(\father(b'))$
and
$\dist(b', \father(b)) \leq \frac{6\sqrt{d}}{\eps}\size(\father(b))$.
Thus
$r \leq \|x-y\|_2 \leq \sqrt{d}\cdot  \size(b) + \sqrt{d}\cdot  \size(b') + \dist(b,b')
 < \sqrt{d}\cdot \size(\father(b)) + \sqrt{d} \cdot \size(\father(b)) + \dist(b',\father(b)) + \sqrt{d} \cdot \size(\father(b)) 
 \leq (3+6/\eps) \sqrt{d}\cdot \size(\father(b))$.
Thus
$\size(\father(b)) > r \eps/(9 \sqrt{d})$.
Let $b_0,\ldots,b_t$ be a branch in $T$ with $b_0=\father(b)$, and $b_t=\{x\}$.
Arguing as in Lemma \ref{lem:D1},
let $t^*\in \{0,\ldots,t-1\}$ be the maximum integer such that $\size(b_{t^*})\geq r\eps /(9\sqrt{d})$.
Let $A\in \Succ(b_{t^*})$, with $b_{t^*+1}\subseteq A$, and pick some box $B(e)\subset A$, with
\begin{align}
\size(B(e)) \in \left[(r\eps)/(18\sqrt{d}), (r\eps)/(9\sqrt{d})\right] \label{eq:Be_size2}
\end{align}
Similarly, let $b_0',\ldots,b_s'$ be a branch of $T$ with $b_0'=\father(b')$, and $b_s'$ is a leaf with $b_s'=\{y\}$.
Arguing as in Lemma \ref{lem:D1}, let $s^*\in \{0,\ldots,s-1\}$ be the maximum integer such that 
$\size(b_{s^*-1}') \geq (r\eps)/(9\sqrt{d})$.
Let $A'\in \Succ(b_{s^*-1})$ with $b_{s^*}\subseteq A$, and pick some box $B'(e)$, with $b_{s^*}\subseteq A \subseteq B'(e)$, such that
\begin{align}
\size(B'(e)) \in \left[(\eps r)/(18\sqrt{d}), (\eps r)/(9\sqrt{d})\right] \label{eq:BBe_size2}
\end{align}
We say that $e$ is \emph{charged} to $\gamma(e)$.
By construction, there exists at most one edge in $E_{1,r}$ that is charged to each pair of boxes.

By \eqref{eq:Be_size2} and \eqref{eq:BBe_size2} we have that for each $e\in E_{2,r}$, the pair $\gamma(e)$ consists of two boxes, each of size $\Theta((\eps r)/\sqrt{d})$.
Thus, each point in $\mathbb{R}^d$ can be contained in at most $O(1)$ distinct boxes in all the pairs $\gamma(e)$, for all $e\in E_{2,r}$.
Moreover, $\dist(B(e),B'(e))\leq \|x-y\| < 2r$.
Thus, each box participates in at most $(\sqrt{d}/\eps)^{O(d)} = (d/\eps)^{O(d)}$ pairs.
For each $e\in E_{2,r}$, let $A(e)=N(B(e),r)\cup N(B'(e),r)$.
It follows that $\{A(e)\}_{e\in E_{2,r}}$ is $(d/\eps)^{O(d)}$-thick.
Since for each $e\in E_{2,r}$, we have $D_e\in A(e)$, it follows that ${\cal D}_2$ is $(d/\eps)^{O(d)}$-thick, as required.
\end{proof}

\begin{lemma}\label{lem:D}
${\cal D}$ is $(2,(d/\eps)^{O(d)})$-thick.
\end{lemma}

\begin{proof}[Proof of Lemma \ref{lem:D}]
By Lemma \ref{lem:D1} we have that ${\cal D}_1$ is $(2,d^{O(d)})$-thick, and by Lemma \ref{lem:D2} we have that ${\cal D}_2$ is $(2,(d/\eps)^{O(d)})$-thick.
Since ${\cal D}={\cal D}_1\cup {\cal D}_2$, we get that ${\cal D}$ is $(2,\kappa)$-thick, where $\kappa=d^{O(d)}+(d/\eps)^{O(d)}=(d/\eps)^{O(d)}$, as required.
\end{proof}


Let $x,y,z,w\in \mathbb{R}^d$.
We say that $zw$ is a \emph{shortcut} for $xy$ if the following conditions holds:

(1)
$\|x-z\|_2 \leq \eps \|z-w\|_2/20$.

(2)
The angle formed by the segments $x$-$y$ and $x$-$(w-z+x)$ is at most $\eps/20$.

We now proceed to modify $G$ to obtain a graph $G'$.
Initially, $G'$ contains no edges.
We consider all edges in $G$ in increasing order of length.
When considering an edge $e=\{x,y\}$, if there exists $\{z,w\}\in E(G')$ such that either $zw$ is a shortcut for $xy$ or $zw$ is a shortcut for $yx$, then we do not add $e$ to $G'$; otherwise we add $e$ to $G'$.
This completes the construction of $G'$. 
We next argue that $G'$ is a spanner with low dilation for $P$.
The proof of the following is standard (see e.g.~\cite{har2011geometric}). For completeness, we provide a sketch of the proof.

\begin{lemma}\label{lem:Gprime_spanner}
$G'$ is a $(1+2\eps)$-spanner for $P$.
\end{lemma}

\begin{proof}
We consider all $\{x,y\}\in {P\choose 2}$ in order of increasing $\|x-y\|_2$ and we prove by induction that $d_{G'}(x,y)\leq (1+2\eps) \|x-y\|_2$.
If $\{x,y\} \notin E(G)$, then the assertion follows by applying the inductive hypothesis on all the edges in the shortest-path between $x$ and $y$ in $G$.
If $\{x,y\} \in E(G)\cap E(G')$, then $d_{G'}(x,y)=\|x-y\|_2$, and the inductive hypothesis holds trivially.
Finally, it remains to consider the case $\{x,y\}\in E(G)\setminus E(G')$.
Since $\{x,y\}$ was not added to $G'$ it follows that there exists some $\{z,w\}\in E(G')$ such that either $zw$ is a shortcut for $xy$ or $zw$ is a shortcut for $yx$.
Assume w.l.o.g.~that $zw$ is a shortcut for $xy$.
We have
\begin{align*}
d_{G'}(x,y) &\leq d_{G'}(x,z)+d_{G'}(z,w)+d_{G'}(w,y)\\
 &\leq (1+2\eps) \|x-z\|_2 + \|z-w\|_2 + (1+2\eps)\|w-y\|_2\\
 &\leq (1+\eps/20+\eps^2/10) \|z-w\|_2 + (1+2\eps) \|w-y\|_2\\
 &< (1+\eps/20+\eps^2/10) (1+\eps/4) (\|x-y\|_2-\|w-y\|_2) + (1+2\eps) \|w-y\|_2\\
 &< (1+\eps) (\|x-y\|_2-\|w-y\|_2) + (1+2\eps)\|w-y\|_2\\
 &< (1+2\eps) \|x-y\|_2,
\end{align*}
which concludes the proof.
\end{proof}

\begin{lemma}\label{lem:pw_cross}
Let $c\in \mathbb{R}^d$ and let $r>0$.
Let 
$E^* = \{\{x,y\} \in E(G') : \|x-y\| > 2r \text{ and }x\text{-}y\cap \sphere(c,r)\neq \emptyset\}$.
Then $|E^*| \leq (d/\eps)^{O(d)}$.
\end{lemma}

\begin{proof}[Proof of Lemma \ref{lem:pw_cross}]
Let $E^*_0=\{\{x,y\}\in E^* : \|x-y\|\leq 100r/\eps\}$.
We can partition $E^*_0$ into $O(\log(1/\eps))$ buckets, where the $i$-th bucket contains the balls with radius in $[r2^i, r2^{i+1})$.
Since by Lemma \ref{lem:D}, ${\cal D}$ is $(2,(d/\eps)^{O(d)})$-thick, and all the balls in $E^*$ are contained in a ball of radius $O(r/\eps)$, it follows that each bucket can contain at most $(1/\eps)^{O(d)} \cdot (d/\eps)^{O(d)}$ balls.
Thus
$|E^*_0| = O(\log(1/\eps)) \cdot (1/\eps)^{O(d)} \cdot (d/\eps)^{O(d)} = (d/\eps)^{O(d)}$.

Let $E_1^*=E^*\setminus E_0^*$.
Suppose that $|E_1^*| > (d/\eps)^{Cd}$.
Setting $C$ to be a sufficiently large universal constant it follows that there exist distinct edges  $\{x,y\}, \{z,w\} \in E_1^*$ that form an angle of less than $\eps/20$.
Assume w.l.o.g.~that $\|x-y\|_2\geq \|z-w\|_2$, $x\in \ball(c,r)$, and $z\in \ball(c,r)$.
Then $zw$ must be a shortcut for $xy$, which is a contradiction since $\{x,y\}\in E(G')$, concluding the proof.
\end{proof}


We now prove the main result of this section.

\begin{theorem}
Let $d\geq 2$ be some fixed integer, and let $\delta \in (0,d]$ be some real number.
Let $P\subset \mathbb{R}^d$ be some finite point set with $|P|=n$, such that $\dimF(P)=\delta$.
Then, for any fixed 
$\eps \in (0,1]$, there exists a $(1+\eps)$-spanner, $G'$, for $P$, with a linear number of edges, and with 
\[
\pw(G)=\left\{\begin{array}{ll}
O(n^{1-1/\delta} \log n) & \text{ if } \delta>1\\
O(\log^2 n) & \text{ if } \delta=1\\
O(\log n) & \text{ if } \delta<1
\end{array}\right.
\]
Moreover, given $P$, the graph $G'$ can be computed in polynomial time.
\end{theorem}

\begin{proof}
Let $G'$ be the spanner constructed above.
The bound on the number of edges of $G'$ follows by Theorem \ref{thm:Vai} since $G'\subseteq G$.
We will bound the pathwidth of $G'$.
By Lemma \ref{lem:D} we have that ${\cal D}$ is $(2,(d/\eps)^{O(d)})$-thick.
By Theorem \ref{sep} there exists some $(d-1)$-sphere $C$ with radius $r$ such that at most $(1-2^{-O(d)})n$ points of $P$ are contained in either side of $C$, and $|A|\leq M$, where
\[
M=\left\{\begin{array}{ll}
O(n^{1-1/\delta}) & \text{ if } \delta>1\\
O(\log n) & \text{ if } \delta=1\\
O(1) & \text{ if } \delta<1
\end{array}\right.,
\]
and $A=\{\{x,y\}\in E(G): \|x-y\|\leq 2r \text{ and } x\text{-}y \cap C \neq \emptyset\}$.

Let $A'=\{\{x,y\}\in E(G'): \|x-y\|\leq 2r \text{ and } x\text{-}y \cap C \neq \emptyset\}$. We have $A'\subseteq A$, and thus $|A'|\leq |A|$.
Let $A''=\{\{x,y\}\in E(G'): \|x-y\|>2r \text{ and } x\text{-}y \cap C \neq \emptyset\}$. By Lemma \ref{lem:pw_cross} we have $|A''|=O(1)$ (for fixed $d$ and $\eps$).
Let $S$ be the set of all endpoints of all the edges in $A'\cup A''$.
We have $|S|\leq 2|A'\cup A''| = O(|A|)$.
Let $U$ (resp.~$U'$) be the set of points in $P$ that are inside (resp.~outside) $C$.
Then $S$ separates in $G'$ every vertex in $U$ from every vertex in $U'$.
We may thus recurse on $G'[U\setminus (S\cup U')]$ and $G'[U'\setminus (S\cup U)]$ and obtain path decompositions $X_1,\ldots,X_t$ and $Y_1,\ldots,Y_s$ respectively.
We now obtain the path decomposition $X_1\cup S, \ldots,X_t\cup S,Y_1\cup S,\ldots,Y_s\cup S$ for $G'$.
The width of the resulting path decomposition is at most $O(M \log n)$, concluding the proof.
\end{proof}

\paragraph{Acknowledgements}
The authors wish to thank Sariel Har-Peled, Dimitrios Thilikos, and Yusu Wang for fruitful discussions.

\bibliography{bibfile}

\newpage
\appendix
\section{Fractal dimension and doubling dimension}\label{sec:doubling}

In this section we observe that the fractal dimension and the doubling dimension of a set of points are related up to a constant factor.

\begin{theorem}\label{thm:doubling}
For any metric space $M$ we have
$\dimD(M) = \dimF(M) + O(1)$ and $\dimF(M) = O(\dimD(M))$.
\end{theorem}

\begin{proof}
Let $M=(X,\rho)$.
We first show that $\dimD(M)=O(\dimF(M))$.
Let $\dimF(M)=\delta$.
For any $\eps$-net $N$ of $M$, at most $O((\frac{r}{\eps})^\delta)$ points in $N$ are contained in any ball of radius $r$, for any $r>\eps$.
Setting $\eps=\frac{r}{2}$ and taking the balls of radius $\frac{r}{2}$ centered at the points of $N$, we get that any ball of radius $r$ is covered by the union of at most $O(2^\delta)$ balls of radius $\frac{r}{2}$.
Thus $\dimD(M)=\delta+O(1)$.

Next we show that $\dimF(M)=O(\dimF(M))$.
Let $\dimD(M)=\lambda$.
From the definition of doubling dimension we have that for any $r>0$, any ball of radius $r$ can be covered by at most $2^\lambda$ balls of radius $\frac{r}{2}$. Given $r>\eps>0$ and $x \in \mathbb{R}^d$, applying the definition of doubling dimension $\log(\frac{2r}{\eps /2})$ times and taking the centers of the balls obtained, we get that there exists $S \subseteq X$ such that $S$ is an $\frac{\eps}{2}$-covering of $X$ and $|S \cap \ball(x,2r)| \leq (\frac{4r}{\eps})^{\lambda}$. Consider any $S' \subseteq X$ such that $S'$ is $\eps$-packing. We have that any two points in $S'$ are covered by different points in $S$ since they are at least distance $\eps$ apart. Also every point in $S' \cap \ball(x,r)$ is covered by some point in $S \cap \ball(x,2r)$. This implies that $|S' \cap \ball(x,r)| \leq |S \cap \ball(x,2r)| \leq (\frac{4r}{\eps})^{\lambda}$. Therefore for any $\eps$-net $N$, which is also an $\eps$-packing by definition, we have that $|N \cap \ball(x,r)| \leq (\frac{4r}{\eps})^{\lambda}$. 
Thus $\dimF(M)=O(\lambda)$, concluding the proof.
\end{proof}

\end{document}